\documentclass{lmcs} 
\pdfoutput=1

\usepackage{lastpage}
\lmcsdoi{16}{1}{13}
\lmcsheading{}{\pageref{LastPage}}{}{}%
{May~01,~2019}{Feb.~13,~2020}{}


\usepackage{hyperref}
\usepackage[english]{babel}
\usepackage[latin1]{inputenc}
\usepackage{amsmath}
\usepackage{amssymb}
\usepackage{latexsym}
\usepackage{amsthm}
\usepackage{indentfirst}
\usepackage{bussproofs}
\usepackage{tikz-cd}
\usepackage[all]{xy}

\theoremstyle{plain}

\theoremstyle{plain}

\newtheorem{lemma}[thm]{Lemma}
\newtheorem{proposition}[thm]{Proposition}
\newtheorem{corollary}[thm]{Corollary}

\theoremstyle{definition}
\newtheorem{definition}[thm]{Definition}
\newtheorem{remark}[thm]{Remark}

\def\sub{\subseteq}
\def\olap{>\mkern-13.5mu <}
\def\P{\mathrm{Pow}}
\def\Pos{\mathrm{Pos}}
\def\cl{\mathrm{cl}}
\def\inte{\mathrm{int}}

\begin{document}

\title[Overlap Algebras]{Overlap Algebras:\\
a Constructive Look at Complete Boolean Algebras} 
\titlecomment{{\lsuper*}This project has received funding from the European Union's Horizon 2020 research and innovation programme under the Marie Sk\l{}odowska-Curie grant agreement No 731143}

\author[F.~Ciraulo]{Francesco Ciraulo}
\address{Department of Mathematics, University of Padova}
\email{ciraulo@math.unipd.it} 

\author[M.~Contente]{Michele Contente}
\address{Scuola Normale Superiore, Pisa}	
\email{michele.contente@sns.it}

\begin{abstract}
  \noindent The notion of a complete Boolean algebra, although completely legitimate in constructive mathematics, fails to capture some natural structures such as the lattice of subsets of a given set. Sambin's notion of an overlap algebra, although classically equivalent to that of a complete Boolean algebra, has powersets and other natural structures as instances.\\ In this paper we study the category of overlap algebras as an extension of the category of sets and relations, and we establish some basic facts about mono-epi-isomorphisms and (co)limits; here a morphism is a symmetrizable function (with classical logic this is just a function which preserves joins). Then we specialize to the case of morphisms which preserve also finite meets: classically, this is the usual category of complete Boolean algebras. Finally, we connect overlap algebras with locales, and their morphisms with open maps between locales, thus obtaining constructive versions of some results about Boolean locales.
\end{abstract}

\maketitle

\section*{Introduction}

A typical phenomenon in constructive mathematics is the split of classical notions: several definitions which are equivalent over classical logic can become deeply different over intuitionistic logic. In this paper we study an alternative way to define complete Boolean algebras, as proposed by Giovanni Sambin \cite{6,3z} who named them \emph{overlap algebras}. There are some facts which make overlap algebras interesting, we believe, from the constructive point of view; for instance, the collection of all subsets of a set is an overlap algebra, actually an atomistic one, although it cannot ever be Boolean (apart from the trivial case of the power of the empty set). 

Roughly speaking, an overlap algebra is a complete lattice (actually a complete Heyting algebra) equipped with a new primitive relation, the overlap relation $\olap$. The intended meaning of $x\olap y$ is that the infimum $x\wedge y$ is ``inhabited''. The distinction between $\emph{inhabited}$ and $\emph{non-empty}$ is enlightening. Indeed, constructively $\exists x(x\in X)$ is a stronger statement than $\neg \forall x \neg(x\in X)$. In an arbitrary complete Heyting algebra we can use $x\neq 0$ as the algebraic counterpart of the set-theoretic $X\neq\emptyset$, but there is no way to express the positive statement of being inhabited. Overlap algebras give an elegant answer to this question. 

Overlap algebras and complete Boolean algebras have just one element in common, the trivial one-element algebra, unless classical logic is assumed, in which case the two notions coincide.

In this paper we investigate two natural notions of morphism between overlap algebras which are both inspired by the powerset construction. First we study the category {\bf OA} as originally introduced by Sambin; {\bf OA} is a dagger category which contains the category {\bf Rel} of sets and relations as a full subcategory; classically, {\bf OA} is the category of complete Boolean algebras and join preserving maps. In particular, we characterize monomorphisms, epimorphisms, and isomorphisms in {\bf OA}, and we establish some basic facts about limits and colimits. 
Then we specialize to the subcategory {\bf OFrm} whose arrows preserve also finite meets. This is a subcategory of {\bf Frm}, the category of frames; morphisms in {\bf OFrm} correspond to open maps in the sense of locale theory. Classically, {\bf OFrm} is the usual category of complete Boolean algebras; we are therefore able to obtain new constructive versions of some standard results about Boolean locales.

If not otherwise stated, we assume to work over intutionistic logic and without choice. In other words, we understand ``constructive'' as ``topos-valid''. In particular, we shall usually think of powersets as perfectly legitimate sets even if we shall make some remark on predicativity in the last section: it is a fact that most of the paper could be adapted to a predicative framework (such as that presented in \cite{4}) by a systematic use of ``bases".

Part of the material in the present paper appeared in the second author's master thesis~\cite{3a}.

\section{Atomic Heyting algebras}

Given a set $X$, its subsets form a complete Heyting algebra $\P(X)$ with respect to the usual set-theoretic operations. Here we write $-Y$ for the pseudo-complement of the subset $Y\subseteq X$. We write $\Omega$ for $\P(1)$, where $1=\{0\}$, which we interpret as the type of truth values. It is well-known that the following statements are equivalent:
\begin{itemize}
    \item the Law of Excluded Middle (LEM);
    \item $(\forall p\in\Omega)(p\cup-p=1)$, that is, $\Omega\cong 2=\{0,1\}$;
    \item $\P(X)$ is a complete Boolean algebra for every $X$.
\end{itemize}

Classically, powersets are precisely the atomic Boolean algebras (this means that every element is the join of the atoms below it, where an atom is a minimal non-zero element). In other words, a Boolean algebra is atomic
 if and only if it is isomorphic to the powerset of the set of its atoms.
The problem of finding a constructive characterization of powersets is related to the problem of finding a suitable algebraization of the notion of a singleton. Apparently, none of the first-order (in the sense of the language of lattices) attempts to define when $a\in L$ is an atom is satisfactory from an intuitionistic point of view; consider, for instance, the following.
\begin{eqnarray}
a\neq 0 & \land & (\forall x\in L)(x\neq 0\land x\leq a\Rightarrow x=a) 
\label{eq def atom 1}\\
a\neq 0 & \land & (\forall x\in L)(x\leq a\Rightarrow x=0\vee x=a) 
\label{eq def atom 2}\\
a\neq 0 & \land & (\forall x\in L)(x< a\Rightarrow x=0) 
\label{eq def atom 3}\\
a\neq 0 & \land & \neg(\exists x\in L)(x\neq 0\land x< a) 
\label{eq def atom 4}
\end{eqnarray}
Indeed, when applied to the case $L=\P(X)$, singletons cannot be proven to be atoms in the sense of \eqref{eq def atom 1} or \eqref{eq def atom 2},  and it is impossible to prove that every subset satisfying \eqref{eq def atom 3} or \eqref{eq def atom 4} is a singleton, although a singleton satisfies \eqref{eq def atom 3} and \eqref{eq def atom 4}. All this comes up clear already when inspecting the case $L=\Omega$: its only singleton $1=\{0\}$ satisfies \eqref{eq def atom 1} or \eqref{eq def atom 2} if and only if LEM holds; and LEM is equivalent to assuming that $1$ is the only $a\in\Omega$ satisfying \eqref{eq def atom 3} or \eqref{eq def atom 4}.
A possible well-known solution is to adopt a second-order definition, as follows. 

\begin{definition}
Given a poset $(L,\leq)$, we say that $a\in L$ is an {\bf atom} if the poset ${\downarrow a}$ = $\{x\in L\ |\ x\leq a\}$ is order-isomoprhic to $\Omega$. And $(L,\leq)$ is {\bf atomic} if the join of all atoms below a given $x$ exists and equals $x$, for every $x\in L$.
\end{definition}

If $L=\P(X)$, then ${\downarrow\{x\}}$ = $\P(\{x\})$ is isomorphic to $\Omega$ = $\P(\{0\})$, for all $x\in X$. So every singleton is an atom and hence every element is a join of atoms. Actually, we can show that the atoms in $\P(X)$ are precisely the singletons. Let $Y$ be an atom, that is, ${\downarrow Y}=\P(Y)\cong\Omega$; and let $\varphi:\P(Y)\to\Omega$ be an order isomorphism (which then preserves joins and meets). Then $1$ = $\varphi(Y)$ = $\varphi(\bigcup_{x\in Y}\{x\})$ = $\bigvee_{x\in Y}\varphi(\{x\})$. So $Y$ is inhabited, actually there is some $x\in Y$ with $\varphi(\{x\})$ = 1, and hence $Y$ = $\varphi^{-1}(1)$ = $\{x\}$.

\begin{proposition}\label{prop atomic frames}
A frame $L$ is atomic if and only if it is order isomorphic to $\P(X)$, where $X$ is the set of atoms of $L$.
\end{proposition}
\begin{proof}
One direction follows from the discussion above.
As for the other, let us define $f:L\to\P(X)$ to be the function which maps a given $x$ to the set of atoms below it, and let $g:\P(X)\to L$ be the function which maps a set of atoms to its join. The two mapping are clearly monotone. Moreover, $g(f(x))=x$ because $L$ is atomic. It remains to show that $f(g(Y))=Y$ for every $Y\subseteq X$. The inclusion $Y\subseteq f(g(Y))$ is clear. As for the other, we must show that $x\leq\bigvee Y$ implies $x\in Y$ for every $x\in X$. Now $x\leq\bigvee Y$ can be written as $x$ = $x\wedge\bigvee Y$ = $\bigvee\{x\wedge y\ |\ y\in Y\}$. Since $x$ is an atom (that is, ${\downarrow x}$ behaves like $\Omega$), $x\wedge y$ must be $x$ for some $y$. So $x\leq y$. Since $y$ is an atom too, this happens precisely when $x=y$ (there is only one atom in ${\downarrow y}$ $\cong$ $\Omega$ = $\P(1)$).
\end{proof} 

\subsection{The positivity predicate on a frame}

For $X$ a set, the statement ``$X$ is inhabited" is stronger than ``$X\neq\emptyset$", constructively; and the two statements are equivalent for all sets $X$ if and only if LEM holds.
There exists a quite standard way to ``algebraize" the concept of an inhabited set. 

\begin{definition}
Let $L$ be a complete lattice. A unary predicate $\Pos$ on $L$ is a {\bf positivity predicate} if the following conditions hold identically. 
\begin{eqnarray}
\Pos (x)\land (x\leq y)\Rightarrow\Pos (y)\label{eq monotonicity}\\
\Pos (\bigvee X)\Rightarrow(\exists x\in X)\Pos (x)\label{eq splitting}\\
(\Pos (x) \Rightarrow (x \leq y)) \Longrightarrow x \leq y
\label{eq positivity}
\end{eqnarray}
\end{definition}

\noindent It is easy to check that \eqref{eq positivity} can be replaced by
\begin{equation}\label{eq positivity bis}
y\leq\bigvee\{x\in L\ |\ \Pos(x)\land (x\leq y)\}\ .    
\end{equation} By extending the terminology which is used for frames/locales, we call a complete lattice {\bf overt} if it has a positivity predicate. 

It is well-known that if $L$ is overt, then $\Pos$ is equivalent to the second-order predicate $POS$, where $POS(x)$ is $(\forall X\subseteq L)(x\leq\bigvee X\Longrightarrow X\textrm{ is inhabited})$. This has a couple of (almost) immediate consequences. First, the positivity predicate, when it exists, is unique and it is uniquely determined by the ordering. Second, $L$ is overt if and only if $POS$ is a positivity predicate.
Classically, every complete lattice is overt and $\Pos(x)$ is just $x\neq 0$. Constructively, $\Pos(x)$ always implies $x\neq 0$, but not the other way around, in general; and it cannot be proven that every complete lattice is overt.

The notion of overteness for a frame can be characterized in a more categorical fashion. Given a frame $L$, there is a unique frame homomorphism $!:\Omega\to L$ (that is, $\Omega$ is the initial frame, that is, the terminal locale). Then $L$ is overt precisely when $!$ has a left adjoint $\exists_!$ (which happens precisely when $!$ preserves arbitrary meets), in which case $\exists_!=\Pos$. 
 
\subsection{Atoms of an overt frame} 

The positivity predicate $\Pos$ can be used to characterize the atoms.
In the case of a powerset, a singleton is precisely a minimal \emph{inhabited} subset. So the following variation of \eqref{eq def atom 1} is the natural candidate for a first-order definition of an atom:
\begin{equation}\label{eq def atom 1bis}
\Pos(a) \land (\forall x\in L)\big(\Pos(x)\wedge (x\leq a)\Longrightarrow (x=a)\big)\ .
\end{equation}

\begin{proposition}
Let $L$ be an overt complete lattice; $a\in L$ is an atom if and only if $a$ satisfies \eqref{eq def atom 1bis}.
\end{proposition}
\begin{proof}
If $L$ is overt and $x\in L$, then also ${\downarrow x}$ is overt with respect to the restriction of $\Pos$.

Let $a$ be an atom, that is, ${\downarrow a}$ $\cong$ $\Omega$. So $\Pos(x)$ becomes ``$x$ is inhabited" under such an isomorphism, and hence \eqref{eq def atom 1bis} is true on ${\downarrow a}$ (because it is true on $\Omega$; recall that the positivity predicate is uniquely determined by the ordering and so has to be preserved by order-isomorphism). 

Conversely, if $a$ satisfies \eqref{eq def atom 1bis}, $\Pos:{\downarrow a}\to\Omega$ is an order-isomorphism whose inverse is $p\mapsto\bigvee\{x\leq a\ |\ p\}$. Indeed, the two mappings are monotone, and $\Pos(\bigvee\{x\leq a\ |\ p\})=p$; moreover, for $b\leq a$, it is $\bigvee\{x\leq a\ |\ \Pos(b)\}\leq b$ because $\Pos(b)$ is just $b=a$ by \eqref{eq def atom 1bis}, and $b\leq\bigvee\{x\leq a\ |\ \Pos(b)\}$ because $\Pos$ is a positivity predicate on ${\downarrow a}$, in particular it satisfies \eqref{eq positivity}.
\end{proof}

As noticed by Giovanni Sambin, \eqref{eq def atom 1bis} is equivalent to the following elegant condition: 
\begin{equation}\label{eq atom Pos}
(\forall x\in L)\big(\Pos(a\wedge x)\Longleftrightarrow(a\leq x)\big)\ .
\end{equation}

\section{Overlap Algebras}

Every complete Boolean algebra is a frame and, classically, every atomic frame (that is, a powerset by proposition \ref{prop atomic frames}) is a complete Boolean algebra. The latter fails constructively; a constructive version can be obtained by replacing complete Boolean algebras by Sambin's {\bf overlap algebras}, as we now see. 

\begin{definition}
An $\emph{overlap-algebra}$ (o-algebra) is an overt frame $L$ such that
\begin{equation}\label{eq prop Pos}  
 (\forall z\in L ) (\Pos (z \wedge x) \Rightarrow \Pos (z \wedge y)) \Longrightarrow x \leq y
\end{equation} for all $x,y\in L$.
\end{definition}
The motivating example is given by powersets, where $\Pos(x)$ means ``$x$ is inhabited". To see that \eqref{eq prop Pos} holds in this case it is sufficient to make $z$ vary over singletons. Note that for $p\in\Omega$ the statement $\Pos(p)$ is equivalent to $p=1$. 

Note that a frame $L$ is an o-algebra if and only if there exists a unary predicate $\Pos$ on $L$ such that \eqref{eq monotonicity}, \eqref{eq splitting} and \eqref{eq prop Pos} hold. Indeed \eqref{eq positivity} follows from \eqref{eq monotonicity} and \eqref{eq prop Pos}: assume $\Pos(x)\Rightarrow (x\leq y)$; if $\Pos(z\wedge x)$, then $\Pos(x)$ and so $x\leq y$ by assumption; therefore $z\wedge x\leq z\wedge y$, and hence $\Pos(z\wedge y)$.

\begin{proposition}
Classically, o-algebras and complete Boolean algebras coincide.
\end{proposition}
\begin{proof}
Classically, overtness is for free, and $\Pos(x)$ is $x\neq 0$. So the implication $\Pos (z \wedge x) \Rightarrow \Pos (z \wedge y)$ in \eqref{eq prop Pos} can be rewritten as $z \wedge y= 0 \Rightarrow z \wedge x= 0$, that is, $z \leq-y \Rightarrow z \leq-x$. Therefore $(\forall z\in L ) (\Pos (z \wedge x) \Rightarrow \Pos (z \wedge y))$ becomes simply $-y\leq-x$ and \eqref{eq prop Pos} becomes $-y\leq-x \Longrightarrow x \leq y$. This holds identically in an Heyting algebra if and only if it is in fact a Boolean algebra.
\end{proof}

Constructively, the previous proposition fails badly, because LEM is equivalent to the statement that $\Omega$ (which is an o-algebra) is Boolean.\footnote{The statement ``every complete Boolean algebra is an o-algebra" is equivalent to LEM as well (see, for instance, \cite{2a} proposition $1.1$).}

Given an o-algebra $L$, it is sometimes convenient to introduce a new relation symbol, say $x\olap y$, for the binary predicate $\Pos(x\wedge y)$: this is the {\bf overlap relation} which gives the name to the structure. If $L$ is a powerset, then $x\olap y$ means that $x$ and $y$ overlap, that is, their intersection is inhabited. Classically, $x\olap y$ is $x\wedge y\neq 0$. Clearly, $\Pos(x)$ is equivalent to $x\olap x$ (and also to $x\olap 1$); this suggests that the definition of an o-algebra can be given in terms of $\olap$ (which was Sambin's original definition).

\begin{proposition}
For $L$ a complete lattice, the following are equivalent:
\begin{enumerate}
    \item $L$ is an o-algebra;
    \item there exists a binary relation $\olap$ on $L$ that satisfies the following properties identically.
        \begin{itemize}
        \item $ x \olap y \Longrightarrow y \olap x $ \hfill $ (\textit{symmetry}) $
        \item $ x \olap y \Longrightarrow x \olap (x \wedge y) $ \hfill $ (\textit{meet closure}) $
        \item $ x \olap \bigvee Y \Longrightarrow (\exists y\in Y) (x \olap y)$ \hfill $ (\textit{splitting of joins} )$
        \item $ (x \olap y)  \land  (y\leq z) \Longrightarrow x \olap z  $  \hfill $ (\textit{monotonicity}) $
        \item $ (\forall z\in L ) (z \olap x \Rightarrow z \olap y) \Longrightarrow x \leq y $ \hfill $ (\textit{density}) $
        \end{itemize}
\end{enumerate}
\end{proposition}
\begin{proof}
The implication $1\Rightarrow 2$ is easy once $x\olap y$ is defined as $\Pos(x\wedge y)$. For instance, splitting of joins holds because binary meets distribute over arbitrary joins (since $L$ is a frame).

As for the reverse implication, we first note that $x\olap y$ is equivalent to $(x\wedge y)\olap (x\wedge y)$ thanks to symmetry, meet closure and monotonicity. Therefore $ (x \wedge y) \olap z$ is always equivalent to $z \olap (y \wedge z)$. We now show that $L$ is a frame, that is, $x\wedge\bigvee Y\leq\bigvee\{x\wedge y\ |\ y\in Y\}$. By density, it is sufficient to check that $z\olap x\wedge\bigvee Y$ implies $z\olap\bigvee\{x\wedge y\ |\ y\in Y\}$. Now $z\olap x\wedge\bigvee Y$ is equivalent to $z\wedge x\olap \bigvee Y$; so there is a $y\in Y$ with $z\wedge x\olap y$, that is, $z\olap x\wedge y$. So $z\olap\bigvee\{x\wedge y\ |\ y\in Y\}$ by monotonicity. Finally, let us define $\Pos(x)$ as $x\olap x$. The only condition on $\Pos$ which needs some proof is \eqref{eq positivity} which follows from \eqref{eq monotonicity} and \eqref{eq prop Pos}, as already noticed. 
\end{proof}

For $L$ an o-algebra, the characterization \eqref{eq atom Pos} of an atom $a\in L$ becomes \begin{equation}\label{eq char atom}
(\forall x\in L)\big((a\olap x)\Longleftrightarrow(a\leq x)\big)\ .
\end{equation} 
By proposition \ref{prop atomic frames}, atomic frames, atomic o-algebras and powersets all amount to the same thing.

\subsection{Non-atomic o-algebras}\label{section non-atomic}

Given any complete Heyting algebra $L$, the set $L_{--}$ = $\{y\in L\ |\ y=--y\}$ has a natural structure of complete Boolean algebra (and every complete Boolean algebra is of this form, because $L_{--}=L$ if $L$ is Boolean). 

A similar result holds for o-algebras \cite{2a}: if $L$ is an overt frame, then the set of all $y\in L$ such that $y$ = $\bigvee\{x\ |\ \forall z(\Pos(z\wedge x)\Rightarrow\Pos(z\wedge y))\}$ is an o-algebra. In particular, if $L=\tau$ where $(X,\tau)$ is a topological space, then we get an o-algebra by considering the set of all $Y\sub X$ such that $Y=\inte\,\cl\,\inte\, Y$, where $\inte$ and $\cl$ are the interior operator and the closure operator corresponding to $\tau$.\footnote{Here $x\in\cl Y$ means that every open neighbourhood of $x$ overlaps $Y$. Assuming that $\cl Y$ is the (set-theoretic pseudo-)complement of the interior of the (set-theoretic pseudo-)complement of $Y$ is tantamount to assuming LEM.} This is a constructive version of the well-known fact that the regular open sets in a topological space form a complete Boolean algebra, which is not atomic, in general, and often with no atoms \cite{3z}.

\section{Morphisms between overlap algebras}

In section \ref{section OFrm} we whall study a category of overlap algebras which, from a classical point of view, is just the category {\bf cBa} of complete Boolean algebras. For the time being, instead, we are going to study a more general kind of morphisms between o-algebras which, classically, correspond to join-preserving maps between complete Boolean algebras.

Sambin's aim in introducing the category $\mathbf{OA}$ of o-algebras was to obtain an extension of the category $ \mathbf{Rel} $ of sets and relations. The definition of an arrow in {\bf OA} makes the assignment $X\mapsto\P(X)$ a functor $\P$ from $ \mathbf{Rel} $ to $ \mathbf{OA} $ which is full, faithful and injective on objects (Proposition \ref{PropPow}).

In the category $\mathbf{Rel}$ a morphism is a binary relation and the composition $S\circ R\subseteq X\times Z$ of the relations $R\subseteq X\times Y$ and $S\subseteq Y\times Z$ is defined by $x(S\circ R)z \Leftrightarrow (\exists y\in Y)(xRy\land ySz)$. 

Given $R\subseteq X\times Y$, its inverse image $R^{-1}:\P(Y)\to\P(X)$ is the function which maps $Y'\subseteq Y$ to $R^{-1}(Y')$ = $\lbrace x\in X\ \vert\ (\exists y\in Y') (xRy)\rbrace$. Clearly, $R^{-1}$ is the identity function on $\P (X)$ if, and only if, $R$ is the equality on $X$.

\begin{lemma}\label{lemmaRel}
For $R\subseteq X\times Y$ and $S\subseteq Y\times Z$, $(S\circ R)^{-1} = R^{-1}\circ S^{-1}$.
\end{lemma}
\begin{proof}
For every $x\in X$ and $D\sub Z$, $x\in (S\circ R)^{-1}(D) $ iff $x(S\circ R)z$ for some $z\in D$; this means that $xRy$ and $ySz$ for some $y\in Y$, and some $z\in D$. In other words, $x\in R^{-1}(S^{-1}(D))$. 
\end{proof}

Each binary relation $R\subseteq X\times Y$ has a ``symmetric'' $R^\dagger\subseteq Y\times X$, where $yR^\dagger x$ iff $xRy$. Its inverse image is a function $(R^\dagger)^{-1}:\P(X)\to\P(Y)$, the direct image of $R$, such that 
\begin{equation}\label{eq.R-}
    R^{-1}(Y')\olap X'\ \textrm{ in }\P(X)\quad\Longleftrightarrow\quad Y'\olap (R^\dagger)^{-1}(X')\ \textrm{ in }\P(Y).
\end{equation} This motivates the following study of symmetrizable functions.

\subsection{Symmetrizable functions}

\begin{definition}
Let $L$ and $M$ be two o-algebras.\footnote{Such a notion makes sense also for $L$ and $M$ overt frames, with $x\olap y$ replaced by $\Pos(x\wedge y)$.} Two functions $f:L\to M$ and $g:M\to L$ are {\bf symmetric} (or {\bf conjugated} \cite{3c}) if
\begin{equation}
f(x) \olap y \Longleftrightarrow x \olap g(y)
\end{equation}
for every $x\in L$ and $y\in M$.\footnote{Classically, the same idea can be expressed by the condition $f(x)\wedge y=0$ $\Leftrightarrow$ $x\wedge g(y)=0$, which is the definition originally proposed \cite{3c}.}
\end{definition}
For instance, the function $a\wedge\_:L\to L$ is self-symmetric, for every element $a$ in an o-algebra $L$.

A function between o-algebras $f:L\to M$ has at most one symmetric.\footnote{This fact fails, in general, when $L$ is an overt frame but not an o-algebra.} Indeed, if $g_1,g_2:M\to L$ are symmetric of $f$, then $ x \olap g_1(y)$ $\Leftrightarrow$ $x \olap g_2(y)$ for every $x$ and $y$; and hence $g_1(y)=g_2(y)$ for every $y$, by density in $L$. 

\begin{definition}
A function $f:L\to M$ between two o-algebras is {\bf symmetrizable} if $f$ has a symmetric. In that case, we write $ f^\dagger$ for the symmetric of $f$.
\end{definition}
Clearly if $f$ is symmetrizable, then $f^\dagger$ is symmetrizable too and $(f^\dagger)^\dagger=f$. Note that if $f$ is symmetrizable, then $f^\dagger$ can be defined in terms of $f$ by means of the formula $f^\dagger(y)$ = $\bigvee \lbrace x \in L\ \vert\ (\forall z \in L) \big(z \olap x \Rightarrow f(z) \olap y\big)\rbrace$.

\begin{proposition}
Let $f:L\to M$ be a function between o-algebras. If $f$ is symmetrizable, then $f$ preserves all joins; and the converse holds classically. 
\end{proposition}
\begin{proof}
For every $y\in M$, we have $y \olap f(\bigvee _i x_{i})$ iff $f^\dagger y \olap \bigvee _ix_{i}$ iff $f^\dagger y\olap x_{i}$ for some $i$ iff $y \olap fx_{i}$ for some $i$ iff $y \olap \bigvee _i fx_{i}$. This shows (by density) that $f(\bigvee _ix_{i}) = \bigvee _i fx_{i}$.

Classically, an o-algebra is exactly a cBa. If $ f: L  \rightarrow  M  $ preserves all joins, then it has a right adjoint $ \forall_f $. We claim that $ f^\dagger  $ does exist and $ f^\dagger(y)  = -\forall_f(-y) $. For $ x \olap -\forall_f(-y)$ $\Leftrightarrow$ $x\wedge -\forall_f(-y) \neq 0$ $\Leftrightarrow$ $x \nleq \forall_f(-y)$ $\Leftrightarrow$ $f(x) \nleq -y$ $\Leftrightarrow$ $f(x) \wedge y \neq 0$ $\Leftrightarrow$ $f(x) \olap y$.
\end{proof}

\begin{remark}\label{remark:symmetrizable}
Classical logic is necessary in the second part of the previous proposition in the sense that LEM follows from the assumption that every join-preserving function between  o-algebras is symmetrizable, as we now see. The argument is based on the fact that LEM is equivalent to assuming that every topological space in which $\cl$ is the identity operator must be discrete.\footnote{\label{counterexample}The following is essentially the same proof given in \cite{2a}. Let us start by constructing a family of topological spaces $(2,\tau_p)$ where $2=\lbrace 0,1\rbrace $ and $p\in\Omega$. Let $\tau_p$ be the topology  of those subsets $X\subseteq 2$ such that if $X$ is inhabited, then either $p$ holds or $p$ implies $X=2$. It is not difficult to check that $\tau_p$ is a topology (and $\tau_p$ is discrete if either $p$ or $\neg p$, which is always the case classically). We claim that every $X\sub 2$ is closed. If $x\in\cl X $, then the open set $\{y\ |\ (y=x)\vee p\}$ must overlap $X$. So either $ x\in X $ or $p$; in the latter case, however, $ \tau_{p} $ is the discrete topology, and hence $ x\in X $ anyway. Therefore $\cl$ is the identity. Now if $\tau_p$ were discrete, then $\lbrace 0\rbrace $ (and $\{1\}$) would be open, hence $p\vee\neg p$ would be true.}
Let us consider any topological space $(X,\tau)$ such that $cl=id$; so $\tau$ is an o-algebra (because every open set is regular). Let $f$ be the inclusion map $\tau\hookrightarrow\P(X)$ and let us assume that $f^\dagger $ exists, that is, $U\olap Y$ $\Leftrightarrow$ $U\olap f^\dagger(Y)$ for every subset $Y$ and every open $U$. This means that $cl(Y)= cl f^\dagger (Y) $ for every $Y$. Since $cl=id$, we get $Y= f^\dagger (Y)$, and hence $Y$ is open, for every $Y$.
\end{remark}

\begin{proposition}\label{prop f from powersets}
Let $f:L\to M$ be a join preserving function between two o-algebras. If $L$ is atomic, then $f$ is symmetrizable.
\end{proposition}
\begin{proof}
Up to order-isomorphism, we can assume that $L$ is $\P(X)$ for some $X$. Put $f^\dagger y=\{x\in X\ |\ f(\{x\})\olap y\}$.
\end{proof}

It is a corollary of the previous proposition (but it can be easily checked directly) that the mapping $X\mapsto\bigvee X$ gives a symmetrizable map from $\P(L)$ to $L$. Its symmetric is given by $y\mapsto\{x\in L\ |\ x\olap y\}$.

\begin{remark}\label{remark f to Omega}
Note that a function $f:L\to\Omega$ is symmetrizable if and only if there is $a\in L$ such that $f(x)=(x\olap a)$ for all $x\in L$. Indeed, given $a\in L$, the mapping $x\mapsto (x\olap a)$ is symmetrizable, and its symmetric maps $p\in\Omega$ to $\bigvee\{x\in L\ |\ (x=a)\land p\}$ (classically, of course, this is just $a$ if $p=1$, and $0$ if $p=0$). 
Conversely, given a symmetrizable $f$, put $a=f^\dagger(1)$. Also note that $a$ is an atom if and only if $f$ preserves finite meets.
\end{remark}

\begin{propC}[{\cite[Theorem $1.15.(iii)$]{3c}}]\label{lemma-ff}
Let $f:L\to M$ and $g:M\to L$ be two functions between o-algebras. Then $f$ and $g$ are symmetric if and only if all the following conditions hold identically:
\begin{enumerate}
    \item $\Pos(f(x))\Rightarrow\Pos(x)$ (classically, $f(0)=0$);
    \item $\Pos(g(y))\Rightarrow\Pos(y)$ (classically, $g(0)=0$);
    \item $f(x)\wedge y\leq f(x\wedge g(y))$;
    \item $x\wedge g(y)\leq g(f(x)\wedge y)$.
\end{enumerate}
\end{propC}
\begin{proof}
Assume that $f$ and $g$ are symmetric. Now $\Pos(f(x))$ can be rewritten as $x\olap g(f(x))$; so $x\olap 1$, that is, $\Pos(x)$; this proves 1, and hence 2 by symmetry.\footnote{\label{remark Posf}Note that $\Pos(fx)$ $\Rightarrow$ $\Pos(x)$ holds true already in the case of overt frames. Indeed, since $x$ = $\bigvee\{x'\ |\ \Pos(x')\land (x'\leq x)\}$ by \eqref{eq positivity bis}, $fx$ = $\bigvee\{fx'\ |\ \Pos(x') \land (x'\leq x)\}$. So if $\Pos(fx)$, then $\Pos(fx')$ for some $x'\leq x$ with $\Pos(x')$; in particular, $\Pos(x')$ for some $x'\leq x$, and hence $\Pos(x)$.} To check 3 (and 4) we use density: $z\olap f(x)\wedge y$ is equivalent to $f(x)\olap z\wedge y$ and hence to $x\olap g(z\wedge y)$; since $g$ is monotone (because it preserves joins), we also have $x\olap g(z)\wedge g(y)$, which is equivalent to $g(z)\olap x\wedge g(y)$ and hence to $z\olap f(x\wedge g(y))$.

Conversely, if $f(x)\olap y$, that is, $\Pos(f(x)\wedge y)$, then also $\Pos(f(x\wedge g(y)))$ by 3; so $\Pos(x\wedge g(y))$ by 1, that is, $x\olap g(y)$; and symmetrically for the other direction. 
\end{proof}

As a corollary, if $f$ is symmetrizable, then $fx$ = $fx\wedge 1$ $\leq$ $f(x\wedge f^\dagger(1))$ $\leq$ $ff^\dagger(fx\wedge 1)$ = $ff^\dagger f x$. And, similarly, $f^\dagger y$ $\leq$ $f^\dagger ff^\dagger y$.

\begin{proposition}\label{prop o-morph}
Let $f:L\to M$ be a symmetrizable function between o-algebras. Then the following conditions are equivalent:
\begin{enumerate}
    \item if $fx_1\olap fx_2$, then $x_1\olap x_2$;
    \item $f$ preserves binary meets;
    \item $f^\dagger fx\leq x$ for every $x$.
\end{enumerate}
Moreover the following are equivalent: 
\begin{enumerate}\setcounter{enumi}{3}
    \item if $x_1\olap x_2$, then $fx_1\olap fx_2$;
    \item if $\Pos(x)$, then $\Pos(fx)$;
    \item $f^\dagger 1=1$;
    \item $x\leq f^\dagger fx$ for every $x$.
\end{enumerate}
\end{proposition}
\begin{proof}

$1\Rightarrow 2$: $z\olap (fx_1\wedge fx_2)$ can be rewritten as $f^\dagger(z\wedge fx_1)\olap x_2$, which yields $(f^\dagger z\wedge f^\dagger fx_1)\olap x_2$; this is equivalent to $fx_1\olap f(f^\dagger z\wedge x_2)$, which implies $x_1\olap (f^\dagger z \wedge x_2)$ by assumption; in other words, $f^\dagger z\olap (x_1\wedge x_2)$, that is,  $z\olap f(x_1\wedge x_2)$.

$2\Rightarrow 3$: $y\olap f^\dagger fx$ iff $fy\olap fx$ iff $(fy\wedge fx)\olap 1$ iff $f(y\wedge x)\olap 1$ iff $(y\wedge x)\olap f^\dagger 1$, and hence $(y\wedge x)\olap 1$, that is, $y\olap x$.

$3\Rightarrow 1$: if $fx_1\olap fx_2$, then $x_1\olap f^\dagger fx_2$ ($\leq x_2$), and hence $x_1\olap x_2$.

$4\Rightarrow 5$: since $\Pos(x)$ is $x\olap x$, and $\Pos(fx)$ is $fx\olap fx$.

$5\Rightarrow 6$: $z\olap 1$ iff $\Pos(z)$, which implies $\Pos(fz)$; this is equivalent to $fz\olap 1$, that is, $z\olap f^\dagger 1$.

$6\Rightarrow 7$: $z\olap x$ iff $(z\wedge x)\olap 1$ iff $(z\wedge x)\olap f^\dagger 1$ iff $f(z\wedge x)\olap 1$; the last implies $(fz\wedge fx)\olap 1$, which is equivalent to $fz\olap fx$ and hence to $z\olap f^\dagger f x$.

$7\Rightarrow 4$: if $x_1\olap x_2$ ($\leq f^\dagger fx_2$), then $x_1\olap f^\dagger fx_2$, that is, $fx_1\olap fx_2$. 
\end{proof}

\section{The category {\bf OA} of overlap algebras}

The identity function $ id: L  \rightarrow L $ on an o-algebra $L$ is symmetrizable with $ id^\dagger = id$. For $L$, $M$, $N$ o-algebras, if $f: L  \rightarrow  M$ and $g:  M  \rightarrow N$ are symmetrizable, then  $ g \circ f $ is symmetrizable too and 
$$ (g\circ f)^\dagger  = f^\dagger  \circ g^\dagger  $$ 
since
$ g(f(x)) \olap z \Leftrightarrow f(x) \olap g^\dagger (z) \Leftrightarrow x \olap f^\dagger (g^\dagger (z)) $. So o-algebras and symmetrizable functions form a category \textbf{OA}. We will sometimes refer to arrows in {\bf OA} as \emph{overlap-morphisms} or \emph{o-morphisms}. 

The category {\bf OA} is a \textit{dagger} category, that is, a category $\mathcal{C}$ equipped with an endofunctor $(\_)^\dagger:\mathcal{C}^{op}\rightarrow \mathcal{C}$ which is the identity on objects and an involution on arrows.

\begin{proposition}\label{PropPow}
Let $\P:\mathbf{Rel}^{op}\rightarrow \mathbf{OA} $ be the functor that associates to every relation $R\subseteq X\times Y$ its inverse image $\P(R)=R^{-1}:\P(Y)\to\P(X)$. Then $\P$ is a full and faithful functor. Moreover, $\P(R^\dagger)=(\P(R))^\dagger$.
\end{proposition}
\begin{proof} The map $\P(R)$ is symmetrizable by equation \eqref{eq.R-}, and $(\P(R))^\dagger=\P(R^\dagger)$. 
Lemma \ref{lemmaRel} shows that $ \P $ is a functor. 

Given $f:\P(Y)\to\P(X)$, let $R\subseteq X\times Y$ be the relation defined as $xRy \Leftrightarrow x\in f(\lbrace y\rbrace)$. Then $x\in R^{-1}(D)$ iff $x\in f(\lbrace y\rbrace)$ for some $y\in D$ iff $x\in \bigcup_{y\in D} f(\lbrace y\rbrace)$ = $f(\bigcup_{y\in D} \lbrace y\rbrace)$ = $f(D)$. This shows that $\P$ is full.
And $\P$ is clearly faithful, for if $ R,S\subseteq X\times Y $ are such that $R^{-1}=S^{-1}$, then $xRy \Leftrightarrow x\in R^{-1}(\lbrace y \rbrace)\Leftrightarrow x\in S^{-1}(\lbrace y \rbrace) \Leftrightarrow xSy$. 
\end{proof}

\subsection{Iso-, mono- and epi-morphisms in OA}

\begin{proposition}
A bijective function $f:L\to M$ is an isomorphism in {\bf OA} if and only if it is an order-isomorphism. In that case $f^\dagger=f^{-1}$ (isomorphisms in {\bf OA} are ``unitary").
\end{proposition}
\begin{proof}
One direction is trivial because all arrows in {\bf OA} are monotone functions.
Conversely, let $f$ be an order-isomorphism; we claim that $f$ is symmetrizable and $f^\dagger=f^{-1}$. As $f$ and $f^{-1}$ preserve binary meets, items 3 and 4 of propositions \ref{lemma-ff} hold; it remains to be shown that $\Pos(f(x))$ implies $\Pos(x)$, and similarly for $f^{-1}$. This follows from the fact that $f$ and $f^{-1}$ preserve joins. Indeed, by \eqref{eq positivity bis}, we have $f(x)$ = $\bigvee\{f(z)\ |\ \Pos(z)\land (z\leq x)\}$. So if $\Pos(f(x))$ holds, then $\Pos(f(z))$ holds for some $z\leq x$ with $\Pos(z)$. In particular, also $\Pos(x)$ holds.
\end{proof}

So $L$ and $M$ are isomorphic in {\bf OA} if and only if they are isomorphic as posets. As a corollary, a join-preserving bijection between o-algebras is always symmetrizable\footnote{In this case, $f$ is an order-isomorphism because $f^{-1}$ preserves joins as well, for $f^{-1}(\bigvee_i y_i)$ = $f^{-1}(\bigvee_i ff^{-1}y_i)$ = $f^{-1}f(\bigvee_i f^{-1}y_i)$ = $\bigvee_i f^{-1}y_i$.} (compare with Remark \ref{remark:symmetrizable}).

\begin{remark}
If $ f$ is an isomorphism in {\bf OA}, then $fx_{1} \olap fx_{2}\Leftrightarrow x_{1} \olap x_{2}$
holds true by Proposition \ref{prop o-morph}.
\end{remark}

\begin{proposition}\label{prop.mono}
Let $ m: L  \rightarrow  M  $ be an arrow in {\bf OA}. Then $m$ is a monomorphism if and only if $m$ is injective.
\end{proposition}
\begin{proof}
If $m$ is an injective function, then it is trivially a monomorphism.

Conversely, assume $ m(a) = m(b) $ with $a,b\in L$. Let $f_a(x)$ and $f_b(x)$ be the truth values of $x\olap a$ and $x\olap b$, respectively. In view of Remark \ref{remark f to Omega}, $f_a$ and $f_b$ are two o-morphisms from $L$ to $\Omega$. Now $f_a m^\dagger y$ = $(m^\dagger y\olap a)$ = $(y\olap ma)$ = $(y\olap mb)$ = $(m^\dagger y\olap b)$ = $f_b m^\dagger y$. So $f_a\circ m^\dagger$ = $f_b\circ m^\dagger$, and hence $m\circ f_a^\dagger$ = $m\circ f_b^\dagger$. Since $m$ is mono, we get $f_a^\dagger$ = $f_b^\dagger$; therefore $f_a$ = $f_b$, that is, $a=b$ (by density).
\end{proof}

Note that a join-preserving map $f:L \rightarrow  M $ between posets is injective if and only if $\forall_f\circ f=id_{L}$ (apply the triangular identity $f\circ\forall_f\circ f=f $). Similarly, $f$ is surjective if and only if $f\circ\forall_f=id_{M}$.

\begin{proposition}
Let $f: L  \rightarrow  M $ be an arrow in {\bf OA}; then
\begin{enumerate}
    \item $f$ is an epimorhism if and only if $f^\dagger$ is a monomorphism;
    \item if $f$ is surjective, then $f$ is an epimorphism;
    \item classically, if $f$ is an epimorphism, then $f$ is surjective.
\end{enumerate}
\end{proposition}
\begin{proof}
Item 1 holds in any dagger category; 2 is trivial.

Let $f$ be epi, and assume LEM. Then $f^\dagger$ is injective, and $f^\dagger(y)$ = $-\forall_f(-y)$. Therefore also $\forall_f$ is injective, for $\forall_f y_1$ = $\forall_f y_2$ iff $-f^\dagger(-y_1)$ = $-f^\dagger(-y_2)$ iff $f^\dagger(-y_1)$ = $f^\dagger(-y_2)$ iff $-y_1$ = $-y_2$ iff $y_1$ = $y_2$. But $\forall_f\circ f\circ\forall_f$ = $\forall_f$ (triangular identity); hence $f(\forall_f y)$ = $y$ for all $y\in M$; so $f$ is surjective.
\end{proof}

It is possible to construct a \emph{Brouwerian} counterexample to the fact that epic implies surjective. Let us consider a topological space $(X, \tau) $ in which the closure operator $\textit{cl}$ is the identity $\textit{id}$ (see section \ref{remark:symmetrizable}). 
Let $f:\P(\tau)\rightarrow\P(X)$ be the map $\lbrace A_{i}\rbrace_{i\in I} \mapsto \cup_{i\in I}A_{i} $; it is symmetrizable and $f^\dagger(Y)$ = $\lbrace A\in \tau\vert Y\olap A\rbrace $.\footnote{This is a consequence of proposition \ref{prop f from powersets}. Here is a direct proof: $f(\lbrace A_{i}\rbrace _{i\in I})\olap Y$ iff $(\bigcup_{i\in I}A_i)\olap Y$ iff $(\exists i\in I)(A_{i}\olap Y)$ iff $(\exists i\in I)(A_{i}\in f^\dagger( Y))$ iff $\lbrace A_{i}\rbrace _{i\in I}\olap f^\dagger( Y)$.} Now $f^\dagger$ is injective, because $f^\dagger(Y_1)$ = $f^\dagger(Y_2)$ iff $(\forall A\in \tau)(Y_1\olap A \Leftrightarrow Y_2\olap A)$ iff $\cl Y_1 = \cl Y_2$ iff $Y_1=Y_2$. In other words, $f^\dagger$ is monic and so $f$ is epic. However, if $f$ were surjective, then every $Y\subseteq X$ would be open. In view of this, if the implication ``$f$ epi $\Rightarrow$ $f$ surjective'' were true, then also ``$\textit{cl}=\textit{id}\Rightarrow \textit{int}=\textit{id}$'' would be true, which is an intuitionistic ``taboo" (see footnote \ref{counterexample} on page \pageref{counterexample}).

\begin{proposition}\label{prop mono-epi}
If $m$ is a mono in {\bf OA}, then the following hold identically:
\begin{enumerate}
    \item if $x_1\olap x_2$, then $m x_1\olap m x_2$;
    \item if $\Pos(x)$, then $\Pos(mx)$;
    \item $m^\dagger 1=1$;
    \item $x\leq m^\dagger m x$.
\end{enumerate}
Symmetrically, If $e$ is an epi in {\bf OA}, then the following hold identically:
\begin{enumerate}
    \item if $y_1\olap y_2$, then $e^\dagger y_1\olap e^\dagger y_2$;
    \item if $\Pos(y)$, then $\Pos(e^\dagger y)$;
    \item $e1=1$;
    \item $y\leq ee^\dagger y$.
\end{enumerate}
\end{proposition}
\begin{proof}
Recall from Proposition \ref{lemma-ff} that $mx\wedge y\leq m(x\wedge m^\dagger y)$ for all $x$ and $y$. In particular, $m1\leq mm^\dagger 1$ and hence $m1=mm^\dagger 1$. If $m$ is a mono, that is, it is injective, then $m^\dagger 1=1$, which is item 6 of Proposition \ref{prop o-morph}. 
The second part follows by applying the same argument to $e^\dagger$.
\end{proof}

\subsection{Limits and co-limits in OA}

Limits and  colimits in a dagger category are mutually closely related: an object $C$ together with arrows $\alpha_i:A_i\to C$ is the colimit of a diagram $f_{i,j}^k: A_i\to A_j$ if and only if the same $C$ together with ${\alpha_i}^\dagger:C\to A_i$ is the limit of the diagram $(f_{i,j}^k)^\dagger: A_j\to A_i$.


\begin{lemma}
Let $\{L_i\}_{i\in I}$ be a family of o-algebras. Then the set-theoretic product $\Pi_{i\in I} L_i$ is an o-algebra with respect to pointwise joins and meets, and $\Pos(f)$ holds in $\Pi_{i\in I} L_i$ if and only if $\Pos(f_i)$ holds in some $L_i$. 
\end{lemma}
\begin{proof}
Let us check \eqref{eq prop Pos}, the other properties being clear. Given $f$ and $g$, assume $\Pos(h\wedge f)\Rightarrow\Pos(h\wedge g)$ for all $h$. For any given $k\in I$ and $z\in L_k$, let us define $h$ as $h_i=\bigvee\{x\in L_i\ |\ (i=k)\land(x=z)\}$. By assumption we then have $\Pos(z\wedge f_k)\Rightarrow\Pos(z\wedge g_k)$, for all $k\in I$ and for all $z\in L_k$. So $f_k\leq g_k$ by \eqref{eq prop Pos} in $L_k$, for all $k\in I$; therefore $f\leq g$.
\end{proof}

\begin{proposition}
The category $ \mathbf{OA}$ has arbitrary products (and coproducts).
\end{proposition}
\begin{proof}
We claim that $\Pi_{i\in I} L_i$, as defined in the previous lemma, is the product of the family of o-algebras $\{L_i\}_{i\in I}$. Let $\pi_k$ be the $k$-th projection, and define ${\pi_k}^\dagger(z)$ to be the function $i\mapsto\bigvee\{x\in L_i\ |\ (i=k)\land(x=z)\}$. Then $f\olap{\pi_k}^\dagger(z)$ iff $f_i\olap \bigvee\{x\in L_i\ |\ (i=k)\land(x=z)\}$ for some $i\in I$ iff $f_k\olap z$ iff $\pi_k(f)\olap z$. Therefore the set-theoretic projections are o-moprhisms.

Let $\{g_i:M\to L_i\}_{i\in I}$ be a family of morphisms in {\bf OA}. We claim that there exists a unique morphism $h:M\to\Pi_{i\in I} L_i$ such that $\pi_i\circ h$ = $f_i$ for all $i$. The only possible candidate for $h$ is the mapping $x\mapsto h(x)$ with $h(x)_i=g_i(x)$. Let us check that it is symmetrizable with  $h^\dagger (f)$ = $\bigvee_{i\in I}{g_i}^\dagger(f_i)$. We have that $h(x)\olap f$ $\Leftrightarrow$ $h(x)_i\olap f_i$ for some $i\in I$ $\Leftrightarrow$ $g_i(x)\olap f_i$ for some $i\in I$ $\Leftrightarrow$ $x\olap {g_i}^\dagger (f_i)$ for some $i\in I$ $\Leftrightarrow x\olap \bigvee_{i\in I}{g_i}^\dagger(f_i)$ $\Leftrightarrow$ $x\olap h^\dagger(f)$. 
\end{proof}

Note that $\P(\emptyset)$ is a zero object (both initial and terminal), because given an arbitrary o-algebra $ L  $, there exists a unique morphism $f: \P(\emptyset)\rightarrow L $, namely, $f(\emptyset) = 0$ ($f$ has to preserve joins); and $f$ is the symmetric of the unique function $g:L\to\P(\emptyset)$, namely, the constant function with value $\emptyset$ (both $\emptyset \olap x$ and $\emptyset \olap gx$ are always false).

\paragraph{The category OA is not complete.}

A category $\mathcal{C}$ is $ \textit{complete} $ if it has all (small) limits. It is well-known that a category with all (small) products is complete if and only if it has equalizers. We are going to show that {\bf OA} does not have equalizers, in general, hence it is not complete. This fact is independent from LEM, that is, {\bf OA} is not complete even classically, as we now see.

Recall that, classically, {\bf OA} is the category of complete Boolean algebras and join-preserving maps. Let us consider the complete Boolean algebras $\Omega=\P(1)\cong 2=\{0,1\}$ and $L=\{0,a,-a,1\}\cong\P(2)$. Let $f,g:L\to 2$ be two maps defined by $f(0)=g(0)=0$, $f(1)=g(1)=1$, $f(a)=g(a)=1$, $f(-a)=0$ and $g(-a)=1$. Clearly both $f$ and $g$ preserves joins. We claim that there is no equalizer of $f$ and $g$. By way of contradiction let us assume that $e:E\to L$ is the equalizer of $f$ and $g$. In particular $e$ is mono, that is, injective; and $-a$ is not in the image of $e$. Therefore, up to isomorphism we have only two possibilities for $E$, namely, $E=1=\P(\emptyset)$, and $E=2$. In particular, the image of $e$ contains at most two elements.
Now consider the function $t:L\to L$ define by $t(0)= 0$, $t(a)= a$, $t(1)= 1=t(-a)$. It is easy to check that $t$ preserves joins and that $f\circ t=g\circ t$. So there must exist (a unique) $h:L \rightarrow E$ such that $e\circ h=t$. This is impossible because the image of $t$ contains three elements.

\begin{remark}
The argument above shows a case in which a $\textit{weak}$ equalizer exists ($t:L\to L$ is a weak equalizer of $f$ and $g$ because any other $h:X\to L$ with $fh=gh$ factors through $t$, actually $h=th$). And weak equalizers always exist in {\bf Rel}.
So it is natural to ask whether {\bf OA} has weak equalizers as well: this is an open problem.
\end{remark}
 
The functor $\P:\mathbf{Rel}^{op}\rightarrow \mathbf{OA} $ preserves (co)products. Indeed, it is well-known that  (co)products in $ \mathbf{Rel} $ are given by disjoint unions; and the powerset of a disjoint union $ \Sigma_{i\in I} X_{i} $ is the set-theoretic product of the powersets of the $X_i$'s, that is, $\P(\Sigma_{i\in I} X_{i})$ = $\Pi_{i\in I} \P(X_{i})$.

\section{Overlap-frames and overlap-locales}\label{section OFrm}

From now on, we restrict our attention to o-morphisms $f$ that preserve finite meets (note that $f^\dagger$ need not preserve finite meets). Let {\bf OFrm} be the corresponding subcategory of {\bf OA}. So {\bf OFrm} is also a subcategory of the category {\bf Frm} of frames, hence the name. Note that the functor $\P$ restricts to a fucntor $\mathbf{Set}^{op}\to\mathbf{OFrm}$ because $R^{-1}$ preserves finite intersections if and only if $R$ is a function.

\begin{lemma}\label{lemma-fff}
Let $ f:L \rightarrow  M  $ be a morphism in {\bf OA}; then the following are equivalent:
\begin{enumerate}
\item $f$ preserves finite meets;
\item $f^\dagger\dashv f$.
\end{enumerate}
\end{lemma}
\begin{proof}
By Proposition \ref{prop o-morph}, $f$ preserves binary meets iff $f^\dagger f x\leq x$, and $f1=1$ iff $y\leq ff^\dagger y$.
\end{proof}

A frame homomorphism $f$ is {\bf open} \cite{3d,johnstone-open} if it has a left adjoint $ \exists_{f} $ which satisfies \emph{Frobenius reciprocity condition} $ \exists_{f} (f(x) \wedge y) = x\wedge\exists_{f}(y)$.
Equivalently, $f$ is open if it preserves the Heyting implication and arbitrary meets. For instance, the unique frame homomorphism $!:\Omega\to L$ is open if and only if the frame $L$ is overt (in which case $\exists_!=\Pos$).

All arrows $f$ in {\bf OFrm} are open with $\exists_f=f^\dagger$; actually the following, more general, result holds.

\begin{proposition}\label{prop morph oFrm}
Let $f:L\to M$ be a function between o-algebras. Then the following are equivalent:
\begin{enumerate}
    \item $f$ is symmetrizable and preserves finite meets;
    \item $f$ is an open frame homomorphism;
    \item $f$ preserves all joins and meets, and the Heyting implication.
\end{enumerate}
\end{proposition}
\begin{proof}
$(1\Rightarrow 2)$: $z\olap f^\dagger (f x\wedge y)$ iff $fz\olap f x\wedge y$ iff $fz\wedge f x\olap y$ iff $f(z\wedge x)\olap y$ iff $z\wedge x\olap f^\dagger y$ iff $z\olap x\wedge f^\dagger y$.

$(2\Rightarrow 3)$: well-known. 

$(3\Rightarrow 1)$: $fx\olap y$ iff $\Pos_M(fx\wedge y)$ iff\footnote{$\Pos_M$ = $\Pos_L\circ\exists_f$ because $!_M$ = $f\circ !_L$.} $\Pos_L\exists_f(fx\wedge y)$ iff $\Pos_L(x\wedge\exists_f y)$ iff $x\olap\exists_f y$.
\end{proof}

So \textbf{OFrm} is the category of o-algebras and open frame homomorphisms. Classically, \textbf{OFrm} is just the category {\bf cBa} of complete Boolean algebras.

\subsection{Subobjects, equalizers and completeness of {\bf OFrm}}

Proposition \ref{prop.mono} holds for {\bf OFrm} too, because the morphisms $f_a^\dagger$ and $f_b^\dagger$ which appear in that proof preserve finite meets. Therefore, every monomorphism $m$ in {\bf OFrm} is injective; by triangular identity, this is equivalent to the equation $m^\dagger\circ m=id$; and this happens precisely $x_1\olap x_2$ $\Leftrightarrow$ $m x_1\olap m x_2$ for all $x_1$, $,x_2$ (see Propositions \ref{prop o-morph} and \ref{prop mono-epi}).

Let $f:L \rightarrow  M $ be any arrow in {\bf OFrm}. Then the set-theoretic image $f[L]$ = $\{f(x)\ |\ x\in L\}$ is a sub-frame of $M$, and it is an o-algebra where  $ fx_{1}\olap_{f[L] } fx_{2}$ is defined as $x_1\olap_L x_2$. Note that the symmetric of the inclusion $ \textit{i}: f[L] \rightarrow  M  $ is given by $ i^\dagger (y)=ff^\dagger (y) $ because
\begin{center}
$ifx\olap_{ M } y\Leftrightarrow fx\olap_{ M } y \Leftrightarrow x\olap_{ L } f^\dagger y\Leftrightarrow  fx\olap_{f[L]} ff^\dagger y$.
\end{center}
Clearly if $m$ is monic, then $m[L]$ is isomorphic to $L$. 

\begin{proposition}
Let $  M  $ be an o-algebra and let $N\subseteq M  $. Then the following are equivalent:
\begin{enumerate}
    \item $N$ = $m[L]$ for some mono $m:L\to M$ in {\bf OFrm};
    \item $N$ is closed under all joins and meets, and the Heyting implication.
\end{enumerate} 
\end{proposition}
\begin{proof}
$(1\Rightarrow 2)$:   easy, since an open frame homomorphism $m$ preserve all joins and meets, and the implication.

$(2\Rightarrow 1)$: by assumption, the inclusion map $ \textit{i}: N \rightarrow  M  $ is an open frame homomorphism, and $\exists_i\circ i$ = $id_N$ because $i$ is injective. We claim that $N$ is an o-algebra, with respect to (the restriction of) the positivity predicate of $M$. The only thing that needs to be checked is \eqref{eq prop Pos}. Given $x,y\in N$,  assume $\Pos(z\wedge x)\Rightarrow \Pos(z\wedge y)$ for all $z\in N$; we must show that $x\leq y$. By \eqref{eq prop Pos} in $M$, it is enough to check that $\Pos(t\wedge x)\Rightarrow \Pos(t\wedge y)$ for all $t\in L$. If $\Pos(t\wedge x)$, then also $\Pos(\exists_i t\wedge x)$ because $\exists_i\vdash i$ and $t\leq i\exists_i t$. By assumption we get $\Pos(\exists_i t\wedge y)$, and hence $\Pos(\exists_i(t\wedge y))$ by Frobenius reciprocity. Since $\exists_i:M\to N$ is a join preserving function between overt frames (see footnote \ref{remark Posf} on page \pageref{remark Posf}), we obtain $\Pos(t\wedge y)$.
\end{proof}

From a classical point of view, of course, $ N$ is a subobject of $  M  $ if and only if  $N$ is a sub-cBa of $M $.

\begin{proposition}
The category {\bf OFrm} is complete.
\end{proposition}
\begin{proof}
The construction of products and equalizers is straightforward. Indeed, if $\{L_i\}_{i\in I}$ is a family of o-algebra, then the product $\Pi_{i\in I}L_i$ in {\bf OA} works as a product in {\bf OFrm} as well (projections $\pi_{i}$'s preserve finite meets). And if $f,g:L\to M$ are two parallel arrows in {\bf OFrm}, then $E=\{x\in L\ |\ fx=gx\}$, together with the inclusion $e:E\to L$, is the equalizer of $f$ and $g$. 
\end{proof}

In general  $\mathbf{OFrm}$ does not have co-products, even classically, because $\mathbf{cBa}$ does not have co-products, in general, as it is well-known. Indeed, this is a consequence of the Gaifman-Hales-Solovay Theorem \cite{solovay} that there is no free complete boolean algebra on countably many generators. 

\subsection{Sublocales of overlap algebras}

Let $\mathbf{Loc}$ = $\mathbf{Frm}^{op}$ be the category of locales (see \cite{3b} and \cite{picado} for a comprehensive treatment of locale theory). A sublocale of $L$ is a regular subobjects in {\bf Loc}, that is, a quotient of $L$ in $\mathbf{Frm}$. It is well known  that sublocales of $L$ have, up to isomorphism, the form $L_{j}$ = ${\lbrace jx\ |\ x\in L \rbrace}$ where $j$ is a {\bf nucleus}, that is, a function $j:L\rightarrow L$ such that
\begin{enumerate}
\item $x\leq j(x) = j(j(x))$ for all $x\in L$, and
\item $j(x\wedge y)= j(x) \wedge j(y)$ for all $x,y\in L$.
\end{enumerate}

By definition, an {\bf open sublocale} is given by a nucleus of the form $j(x)= a\to x$, for $a\in L$, where $\to$ is the Heyting implication in $L$. It is well known \cite{3d} that a sublocale $m:L_j\to L$ is open if and only if the corresponding frame epimorphism $m^*:L\to L_j$, with $m^*(x)=j(x)$, is open. Moreover, $\exists_{m^*}(jx)=a\wedge x$ and $a=\exists_{m^*}(1)$.







A sublocale $L_j$ is Boolean, that is, it is a complete Boolean algebra if and only if $j$ is of the form $j(x)=(x\to a)\to a$ for some $a\in L$. If $L$ itself is Boolean, then every sublocale $L_j$ of $L$ is Boolean, because $j(x)=(x\to j(0))\to j(0)$ holds identically in that case. What happens if we replace complete Boolean algebras with overlap algebras? 

\begin{proposition}
Every open sublocale of an overlap algebra is an overlap algebra.
\end{proposition}
\begin{proof}
Given $m:L_{j}\rightarrow L$ with $L$ an o-algebra and $m$ open, we claim that $L_{j}$ is an o-algebra with respect to the positivity predicate $\Pos_{L_{j}}=\Pos_{L}\circ \exists_{m^{*}}$. 
$$\xymatrix{
L_j\ar@/^/[rr]^{\exists_{m^*}} \ar@<1ex>@/^2pc/[rrrr]^{Pos_{L_j}} & & L \ar@/^/[rr]^{Pos_L} \ar@/^/[ll]^{m^*} & & \Omega\ar@/^/[ll]^{!_l^*} \ar@<1ex>@/^2pc/[llll]^{!_{L_j}^*}
}$$
In order to prove that $L_{j}$ is an o-algebra we must check that \eqref{eq prop Pos} holds for $L_j$, namely 
\begin{center}
$\forall z. \left[\Pos_{L_{j}}(jz\wedge jx)\Rightarrow\Pos_{L_{j}}(jz\wedge jy)\right]    \Longrightarrow jx\leq jy\ .$
\end{center}
Now $\Pos_{L_{j}}(jz\wedge jx)$ can be rewritten as $\Pos_{L}\exists_{m^{*}}(m^*z\wedge jx)$, and hence as $\Pos_{L}(z\wedge\exists_{m^{*}}jx)$; similarly for $y$ in place of $x$. So the antecedent becomes $\forall z. [(z\olap \exists_{m^{*}}jx)\Rightarrow(z\olap \exists_{m^{*}}jy)]$, that is, $\exists_{m^{*}}jx\leq \exists_{m^{*}}jy$. This is equivalent to $jx\leq m^{*}\exists_{m^{*}}jy$ = $m^{*}\exists_{m^{*}}m^*y$ = $m^*y$ = $jy$. 
\end{proof}

Discrete locales, that is, powersets regarded as locales, are overlap algebras (and they are Boolean if and only if LEM holds). More generally, we have the following.

\begin{lemma}
Every overt sublocale of a discrete locale is open (as a sublocale).
\end{lemma}
\begin{proof}
Let $j$ be a nucleus on $\P(X)$ such that the corresponding sublocale is overt with positivity predicate $\Pos$. Let $P$ be $\{x\in X\ |\ \Pos(j\{x\})\}$. We claim that $jU=P\to U$. Indeed, if $x\in jU$, then $j\{x\}=j(\{x\}\cap U)$; if also $x\in P$, then $\Pos(j(\{x\}\cap U))$, and hence $\{x\}\cap U$ is inhabited, that is, $x\in U$. Conversely, if $x\in P\to U$, then $\Pos(j\{x\})\Rightarrow (x\in U)$; so $\Pos(j\{x\})\Rightarrow (j\{x\}\subseteq jU)$; by overtness, $j\{x\}\subseteq jU$, that is, $x\in j U$.
\end{proof}

\begin{corollary}
Overt sublocales of discrete locales are overlap algebras.
\end{corollary}
\begin{proof}
By the previous proposition and lemma.
\end{proof}


\begin{proposition}
For $L$ a locale, there is a bijection between sublocales of $L$ which are overlap algebras and join-preserving maps $L\to\Omega$.
\end{proposition}
\begin{proof}
Given an o-algebra $L_j$, put $\varphi(x)= \Pos(jx)=(jx\olap jx)$. Then $\varphi(\bigvee_i x_i)$ = $\Pos(j(\bigvee_i x_i))$ = $\Pos(\bigvee^{L_j}_i jx_i)$ = $\exists i.\Pos(jx_i)$ = $\exists i.\varphi(x_i)$.

Conversely, given $\varphi :L\to \Omega$, put $jy$ = $\bigvee\{x\in L\ |\ \forall z.[\varphi(z\wedge x)\Rightarrow\varphi(z\wedge y)]\}$. It is not difficult to check that $\varphi(x\wedge jy)$ iff $\varphi(x\wedge y)$, and that $x\leq jy$ iff $\varphi(z\wedge x)\Rightarrow\varphi(z\wedge y)$ for all $z$. Therefore $j$ is a nucleus, and $L_j$ is an o-algebra with $jx\olap jy$ if $\varphi(x\wedge y)$.
\end{proof}

\section*{Some remarks on predicativity}

In predicative foundations powersets are treated essentially as classes; actually, complete (semi)lattices are typically partially ordered classes rather than posets. As a consequence, the requirement \eqref{eq prop Pos} in the very definition of an overlap algebra appears problematic, as it may contain a universal quantification over a class. 

This problem can be often overcome by restricting one's attention to $\textit{set-based}$ overlap algebras. A $ \textit{base} $ $S$ for a suplattice (complete join-semilattice) $ (L , \leq )$ is a set-indexed family of generators:  $p = \bigvee \lbrace a\in S\vert a\leq p\rbrace$ for every $p$ in $L$. Of course, every o-algebra is set-based impredicatively.

For a set-based o-algebra condition \eqref{eq prop Pos} can be replaced by the following 
\begin{center}
$ (\forall a\in S) (\Pos(a\wedge x) \Rightarrow \Pos(a\wedge y)) \Longrightarrow x\leq y $
\end{center} where the universal quantifier ranges over a set now.



Much of the results about {\bf OA} presented here remain valid for the category of set-based o-algebras  within a predicative framework.

\bibliographystyle{plain}
\bibliography{bibliography}

\begin{thebibliography}{10}

\bibitem{2a}
Francesco Ciraulo.
\newblock Overlap algebras as almost discrete locales.
\newblock {\em (Submitted)}, 2016.
\newblock arXiv:1601.04830.

\bibitem{3z}
Francesco Ciraulo and Giovanni Sambin.
\newblock The overlap algebra of regular opens.
\newblock {\em J. Pure Appl. Algebra}, 214(11):1988--1995, 2010.

\bibitem{3a}
Michele Contente.
\newblock Intuitionistic type theory vs minimalist foundation: Theory and
  practice.
\newblock Master's thesis, Universit\`{a} degli Studi di Siena, 2017.

\bibitem{johnstone-open}
Peter~T. Johnstone.
\newblock Open maps of toposes.
\newblock {\em Manuscripta Math.}, 31(1-3):217--247, 1980.

\bibitem{3b}
Peter~T. Johnstone.
\newblock {\em Stone spaces}, volume~3 of {\em Cambridge Studies in Advanced
  Mathematics}.
\newblock Cambridge University Press, Cambridge, 1982.

\bibitem{3c}
Bjarni J\'{o}nsson and Alfred Tarski.
\newblock Boolean algebras with operators. {I}.
\newblock {\em Amer. J. Math.}, 73:891--939, 1951.

\bibitem{3d}
Andr\'{e} Joyal and Myles Tierney.
\newblock An extension of the {G}alois theory of {G}rothendieck.
\newblock {\em Mem. Amer. Math. Soc.}, 51(309):vii+71, 1984.

\bibitem{4}
Maria~Emilia Maietti and Giovanni Sambin.
\newblock Toward a minimalist foundation for constructive mathematics.
\newblock In {\em From sets and types to topology and analysis}, volume~48 of
  {\em Oxford Logic Guides}, pages 91--114. Oxford Univ. Press, Oxford, 2005.

\bibitem{picado}
Jorge Picado and Ale\v{s} Pultr.
\newblock {\em Frames and locales}.
\newblock Frontiers in Mathematics. Birkh\"{a}user/Springer Basel AG, Basel,
  2012.

\bibitem{6}
Giovanni Sambin.
\newblock {\em Positive Topology and the Basic Picture. New structures emerging
  from constructive mathematics}.
\newblock Oxford University Press, Oxford, (forthcoming).

\bibitem{solovay}
Robert~M. Solovay.
\newblock New proof of a theorem of {G}aifman and {H}ales.
\newblock {\em Bull. Amer. Math. Soc.}, 72:282--284, 1966.

\end{thebibliography}

\end{document}